\documentclass[letterpaper, 10 pt, conference]{ieeeconf}  

\IEEEoverridecommandlockouts                              
\usepackage{calc}  
\usepackage{amsmath,amssymb}  
\usepackage{cite}

\usepackage[amsmath,thmmarks]{ntheorem}
\usepackage{amsfonts} 
\usepackage{color}
\usepackage{graphics} 
\usepackage{algorithm}
\usepackage[noend]{algpseudocode}
\usepackage[style=base]{subcaption}
\usepackage{diagbox}
\usepackage[hidelinks]{hyperref}
\usepackage[capitalise]{cleveref}
\usepackage{units}

\crefname{equation}{}{} 
\crefname{section}{Sec.}{Sec.}

\DeclareCaptionLabelSeparator{periodspace}{.\quad}
 \newcommand{\N}{\mathbb{N}}
 
 \newcommand{\R}{\mathbb{R}}
 
 \newcommand{\X}{\mathcal{X}}
 \newcommand{\U}{\mathcal{U}}

\theoremstyle{plain}
\newtheorem{theorem}{Theorem}

\newtheorem{corollary}[theorem]{Corollary}
\newtheorem{definition}{Definition}
\newtheorem{lemma}{Lemma}

\newtheorem{assumption}{Assumption}

\graphicspath{{figures/}}

\title{\LARGE \bf
Learning-based Model Predictive Control for Safe Exploration 
}

\author{Torsten Koller, Felix Berkenkamp, Matteo Turchetta and Andreas Krause
\thanks{This work was supported by SNSF grant {200020\_159557}, a fellowship within the FITweltweit program of the German Academic Exchange Service (DAAD), the Vector Institute, an Open Philantropy Project AI fellowship, and the Max Planck ETH Center for Learning Systems.}
\thanks{Torsten Koller is with the Department of Computer Science, University of Freiburg, Germany. Email: kollert@informatik.uni-freiburg.de}
\thanks{Felix Berkenkamp, Matteo Turchetta and Andreas Krause are with the Learning \& Adaptive Systems Group, Department of Computer Science, ETH Zurich, Switzerland. Email: \{befelix, matteotu, krausea\}@inf.ethz.ch}%
}

\usepackage{fancyhdr}
\newcommand{\mytitle}{\textbf{Accepted at}
\textit{the IEEE Conference on Decision and Control, 2018
}.\\[0.7em]
\copyright 2018 IEEE. Personal use of this material is permitted. Permission from IEEE must be obtained for all other uses, in any current or future media, including reprinting/republishing this material for advertising or promotional purposes, creating new collective works, for resale or redistribution to servers or lists, or reuse of any copyrighted component of this work in other works.}
\begin{document}

\maketitle
\thispagestyle{fancy}
\pagestyle{empty}


\begin{abstract}
Learning-based methods have been successful in solving complex control tasks without significant prior knowledge about the system. However, these methods typically do not provide any safety guarantees, which prevents their use in safety-critical, real-world applications. 
In this paper, we present a learning-based model predictive control scheme that can provide provable high-probability safety guarantees.
To this end, we exploit regularity assumptions on the dynamics in terms of a Gaussian process prior to construct provably accurate confidence intervals on predicted trajectories. Unlike previous approaches, we do not assume that model uncertainties are independent.
Based on these predictions, we guarantee that trajectories satisfy safety constraints. Moreover, we use a terminal set constraint to recursively guarantee the existence of safe control actions at every iteration.
In our experiments, we show that the resulting algorithm can be used to safely and efficiently explore and learn about dynamic systems.
\end{abstract}

\section{Introduction}

 In model-based reinforcement learning (RL,\cite{Sutton1998Reinforcement}), we aim to learn the dynamics of an unknown system from data, and based on the model, derive a policy that optimizes the long-term behavior of the system. Crucial to the success of such methods is the ability to efficiently explore the state space in order to quickly improve our knowledge about the system. While empirically successful, current approaches often use exploratory actions during learning, which lead to unpredictable and possibly unsafe behavior of the system, e.g., in exploration approaches based on the \textit{optimism in the face of uncertainty} principle \cite{Xie2016Modelbased}. Such approaches are not applicable to real-world safety-critical systems. 

In this paper we introduce \textsc{SafeMPC}, a safe model predictive control (MPC) scheme that guarantees the existence of feasible \textit{return trajectories} to a safe region of the state space at every time step with high-probability. These return trajectories are identified through a novel uncertainty propagation method that, in combination with constrained MPC, allows for formal safety guarantees in learning control. 
\begin{figure}
\includegraphics{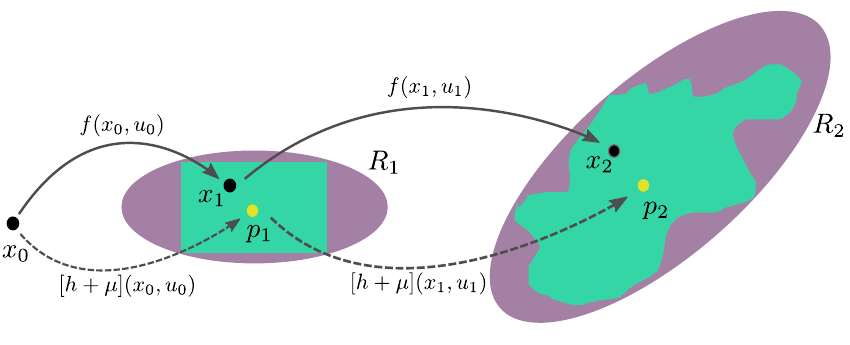}
\caption{Propagation of uncertainty over multiple time steps based on a well-calibrated statistical model of the unknown system. We iteratively compute ellipsoidal over-approximations (purple) of the intractable image (green) of the learned model for uncertain ellipsoidal inputs.
}
\end{figure}

\paragraph*{Related Work}
One area that has considered safety guarantees is robust MPC. There, we iteratively optimize the performance along finite-length trajectories at each time step, based on a known model that incorporates uncertainties and disturbances acting on the system \cite{rawlings2009model}. In a constrained robust MPC setting, we optimize these local trajectories under additional state and control constraints. Safety is typically defined in terms of recursive feasibility and robust constraint satisfaction.
In~\cite{Sadraddini2016Provably}, this definition is used to safely control urban traffic flow, while~\cite{Carson2013Robust} guarantees safety by switching between a standard and a safety mode. However, these methods are conservative since they do not update the model.

In contrast, learning-based control approaches adapt their models online based on observations of the system. This allows the controller to improve over time, given limited prior knowledge of the system. Theoretical safety guarantees in learning-based MPC (LBMPC) are established in \cite{Aswani2013Provably}. A safety mechanism for general learning-based controllers using robust MPC is proposed in~\cite{Wabersich2018Linear}. Both approaches require a known nominal linear model. The former approach requires deviations from the system dynamics to be bounded in an pre-specified polytope, the latter relies on sampling.

MPC based on Gaussian process (GP, \cite{Rasmussen2006Gaussian}) models is proposed in a number of works, e.g. \cite{Kocijan2004Gaussian,Cao2017Gaussian}. The difficulty here is that trajectories have complex dependencies on states and unbounded stochastic uncertainties.
Safety through probabilistic \textit{chance constraints} is considered in \cite{Hewing2018Cautious,jainLearningControlUsing2018,Ostafew2016Robust} based on approximate uncertainty propagation. While often being empirically successful, these approaches do not theoretically guarantee safety of the underlying system.

Another area that has considered learning for control is model-based RL. There, we aim to learn global policies based on data-driven modeling techniques, e.g., by explicitly trading-off between finding locally optimal policies (exploitation) and learning the behavior of the system globally (exploration) ~\cite{Sutton1998Reinforcement}. This results in data-efficient learning of policies in unknown systems~\cite{Deisenroth2011PILCO}. In contrast to MPC, where we optimize finite-length trajectories, in RL we typically aim to find an infinite horizon optimal policy. Hence, enforcing hard constraints in RL is challenging. Control-theoretic safety properties such as Lyapunov stability or robust constraint satisfaction are only considered in a few works~\cite{Ernst2009Reinforcement}. In~\cite{Berkenkamp2017Safe-1}, safety is guaranteed by optimizing parametric policies under stability constraints, while~\cite{Akametalu2014Reachabilitybased} guarantees safety in terms of constraint satisfaction through reachability analysis.
\paragraph*{Our Contribution}
We combine ideas from robust control and GP-based RL to design a MPC scheme that recursively guarantees the existence of a safety trajectory that satisfies the  constraints of the system. In contrast to previous approaches, we use a novel uncertainty propagation technique that can reliably propagate the confidence intervals of a GP-model forward in time. We use results from statistical learning theory to guarantee that these trajectories contain the system with high probability jointly for all time steps. In combination with a constrained MPC approach and a terminal set constraint, we then prove the safety of the system.
We apply the algorithm to safely explore the dynamics of an inverted pendulum simulation.

\section{Problem Statement}
\label{ps}
We consider a nonlinear, discrete-time dynamical system
\begin{equation}
\label{ps:eq:system}
x_{t+1} = f(x_t,u_t) = \underbrace{h(x_t,u_t)}_{\text{prior model}} + \underbrace{g(x_t,u_t)}_{\text{unknown error}},
\end{equation}
where $x_t \in \R^{n_x}$ is the state and $u_t \in \R^{n_u}$ is the control input to the system at time step $t\in \N$. We assume that we have access to a twice continuously differentiable prior model  $h(x_t,u_t)$, which could be based on a first principles physics model. The model error $g(x_t,u_t)$ is \textit{a priori} unknown and we use a statistical model to learn it by collecting observations from the system during operation. In order to provide guarantees, we need reliable estimates of the model-error. In general, this is impossible for arbitrary functions~$g$. We make the following additional regularity assumptions.

We assume that the model-error $g$ is of the form $g(z) = \sum_{l=0}^\infty \alpha_i k(z,z_l), \, \alpha_l \in \R, z = (x,u) \in \R^{n_x} \times \R^{n_u}  $, a weighted sum of distances between inputs $z$  and representer points $z_l = (x_l,u_l) \in \R^{n_x} \times \R^{n_u}$ as defined through a symmetric, positive definite \textit{kernel} $k$. This class of functions is well-behaved in the sense that they form a reproducing kernel Hilbert space (RKHS, \cite{wahba1990spline}) $\mathcal{H}_k$ equipped with an inner-product $\langle \cdot,\cdot \rangle_k$. The induced norm $||g||_k^2 = \langle g,g \rangle_k$ is a measure of the \textit{complexity} of a function $g \in \mathcal{H}_k$. Consequently, the following assumption can be interpreted as a requirement on the smoothness of the model-error $g$ w.r.t. the kernel $k$.
\begin{assumption}
\label{ass:rkhs}
The unknown function $g$ has bounded norm in the RKHS $\mathcal{H}_k$, induced by the  continuously differentiable kernel $k$, i.e. $||g||_k \leq B_g$.
\end{assumption}
 In the case of a multi-dimensional output $n_x > 1$, we follow \cite{Berkenkamp2016Bayesian} and redefine $g$ as a single-output function $\tilde{g}$ such that $\tilde{g}(\cdot,j) = g_j(\cdot)$ and assume that $||\tilde{g}||_k \leq B_g$.

We further assume that the system is subject to polytopic state and control constraints
\begin{align}
\label{ps:ass:state_constraits}
\X = \{x \in \R^{n_x} | H_x x \leq h_x, \, h_x \in \R^{m_x} \}, \\
\label{ps:ass:ctrl_constraits}
\U = \{u \in \R^{n_u} | H_u u \leq h_u, \, h_u \in \R^{m_u} \},
\end{align}
which are bounded. For example, in an autonomous driving scenario, the  state region could correspond to a highway lane and the control constraints could represent the physical limits on acceleration and steering angle of the car.

Lastly, we assume access to a backup controller that guarantees that we remain inside a given safe subset of the state space once we enter it. In the autonomous driving example, this could be a simple linear controller that stabilizes the car in a small region in the center of the lane
at slow speeds.

\begin{assumption}
\label{ps:ass:safe_set}
We are given a controller $\pi_{\mathrm{safe}}(\cdot)$ and a polytopic safe region
\begin{equation}
\X_{\mathrm{safe}} := \{ x\in \R^{n_x} | H_s x \leq h_s\} \subseteq \X,
\end{equation}
which is (robust) control positive invariant (RCPI) under~$\pi_{\mathrm{safe}}(\cdot)$. Moreover, the controller satisfies the control constraints inside~$\X_\mathrm{safe}$, i.e. $\pi_{\mathrm{safe}}(x) \in \mathcal{U} \, \forall x \in \X_{\mathrm{safe}}$.
\end{assumption}
This assumption allows us to gather initial data from the system inside the safe region even in the presence of significant model errors, since the system remains safe under the controller~$\pi_\mathrm{safe}$. Moreover, we can still guarantee constraint satisfaction asymptotically outside of~$\X_\mathrm{safe}$, if we can show that a finite sequence of control inputs eventually steers the system back to the safe set~$\X_\mathrm{safe}$. This idea and a similar definition of a safe set was introduced concurrently in \cite{Wabersich2018Linear}.
A set and corresponding controller which fulfill \cref{ps:ass:safe_set} for general dynamical systems is difficult to find. However, there has been recent progress in finding stability regions for systems of the form \cref{ps:eq:system}, which are RCPI by design, that could, under additional considerations (e.g. through polytopic inner-approximations ~\cite{BronsteinApproximationconvexsets2008}), satisfy the assumptions.

Given a controller $\pi$, ideally we want to enforce the state- and control constraints at every time step,
\begin{equation}
\label{ps:eq:safety_deterministic}
\forall t \in \N: f_\pi(x_t) \in \mathcal{X}, \, \pi(x_t) \in \mathcal{U},
\end{equation}
where $x_{t+1} = f_\pi(x_t) = f(x_t,\pi(x_t))$ denotes the closed-loop system under~$\pi$. Apart from $\pi_{\mathrm{safe}}$, which trivially and conservatively fulfills this, it is in general impossible to design a controller that enforces
\cref{ps:eq:safety_deterministic} without additional assumptions. Instead, we slightly relax this requirement to \textit{safety with high probability} throughout its operation time.
\begin{definition}
\label{ps:def:epsilon_safety}
Let $\pi: \R^{n_x} \to \R^{n_u}$ be a controller for \cref{ps:eq:system} with the corresponding closed-loop system $f_\pi$.
Let $x_0 \in \mathcal{X}$ and $\delta \in (0,1]$. A system is $\delta-$safe under the controller $\pi$ iff:
\begin{equation}
\Pr\left[ \, \forall t \in \N:  f_\pi(x_t) \in \mathcal{X}, \, \pi(x_t) \in \mathcal{U}  \right] \geq 1-\delta.
\end{equation}

\end{definition}

Based on~\cref{ps:def:epsilon_safety}, the goal is to design a control scheme that guarantees $\delta$-safety of the system~\cref{ps:eq:system}. At the same time, we want to improve our model by learning from observations collected outside of the initial safe set $\mathcal{X}_\mathrm{safe}$ during operation, which increase the performance of the controller over time. 

\section{Background}
In this section, we introduce the necessary background on GPs and set-theoretic properties of ellipsoids that we use to model our system and perform multi-step ahead predictions.
\subsection{Gaussian Processes (GPs)}
We want to learn the unknown model-error $g$ from data using a GP model.
%
A $\mathcal{GP}(m,k)$ is a distribution over functions, which is fully specified through a mean function~$m: \R^{d} \to \R$ and a \textit{covariance function} $k: \R^d \times \R^d \to \R$, where $d = n_x + n_u$. Given a set of $n$ noisy observations $y_i = f(z_i) + w_i, \, w_i \sim \mathcal{N}(0,\lambda^2), i = 1,\dots,n, \, \lambda \in \R$, we choose a zero-mean prior on $g$ as $m \equiv 0$ and regard the differences $\tilde{y}_n = [y_1 - h(z_1),\dots,y_n - h(z_n)]^T$ between prior model $h$ and observed system response at input locations $Z = [z_1,..,z_n]^T$. The posterior distribution at~$z$ is then given as a Gaussian $\mathcal{N}(\mu_n(z),\sigma_n^2(z))$
with mean and variance
\begin{align}
\label{found:gp:pred_mu}
\mu_n(z) &= k_n(z)^\mathrm{T} [K_n+\lambda^2 I_n]^{-1}\tilde{y}_n \\
\label{found:gp:pred_var}
\sigma^2_n(z) &= k(z,z) - k_n(z)^\mathrm{T} [K_n + \lambda^2 I_n]^{-1}k_n(z),
\end{align}
where $[K_n]_{ij} = k(z_i,z_j), [k_n(z)]_j = k(z,z_j)$, and $I_n$ is the $n-$dimensional identity matrix.
In the case of multiple outputs $n_x > 1$, we model each output dimension with an independent GP, $\mathcal{GP}(m_j,k_j), j = 1,..,n_x$. We then redefine \cref{found:gp:pred_mu} and \cref{found:gp:pred_var} as $\mu_n(\cdot) = (\mu_{n,1}(\cdot),..,\mu_{n,n_x}(\cdot))$ and $\sigma_n(\cdot) = (\sigma_{n,1}(\cdot),..,\sigma_{n,n_x}(\cdot))$ corresponding to the predictive mean and variance functions of the individual models.


Based on~\cref{ass:rkhs}, we can use GPs to model the unknown part of the system~\cref{ps:eq:system}, which provides us with reliable confidence intervals on the model-error $g$.
\begin{lemma}
\label{gp:lemma:rkhs_confidence_ivals}
\cite[Lemma 2]{Berkenkamp2017Safe-1}: Assume $||g||_k \leq B_g$ and that measurements are corrupted by $\lambda$-sub-Gaussian noise. Let $\beta_n = B_g + 4\lambda \sqrt{\gamma_n + 1 + \ln(1/\delta)}$, where $\gamma_n$ is the information capacity associated with the kernel $k$. Then with probability at least $1-\delta$ we have for all $ 1\leq j \leq n_x,\, z \in \mathcal{X} \times \mathcal{U}$ that $|\mu_{n-1,j}(z) - g_j(z)| \leq \beta_n \cdot \sigma_{n-1,j}(z)$.
\end{lemma}

In combination with the prior model $h(z)$, this allows us to construct reliable confidence intervals around the true dynamics of the system \cref{ps:eq:system}. The scaling $\beta_n$ depends on the number of data points~$n$ that we gather from the system through the information capacity,
$\gamma_n = \max_{A \subset \tilde{Z} , |A| = \tilde{n}} I(\tilde{g}_A;g), \,  \tilde{Z} = \mathcal{X} \times \mathcal{U} \times \mathcal{I}, \, \tilde{n} = n \cdot n_x$, i.e. the maximum mutual information $I(\tilde{g}_A,g)$ between a finite set of samples $A$ and the function $g$.
Exact evaluation of $\gamma_n$ is NP-hard in general, but it can be greedily approximated and has sublinear dependence on $n$ for many commonly used kernels~\cite{Srinivas2010Gaussian}.

The regularity assumption~\cref{ass:rkhs} on our model-error and the smoothness assumption on the covariance function $k$ additionally imply that the function $g$ is Lipschitz.
\subsection{Ellipsoids}
We use ellipsoids to give an outer bound on the uncertainty of our system when making multi-step ahead predictions.
Due to appealing geometric properties, ellipsoids are widely used in the robust control community to compute \textit{reachable sets} \cite{Filippova2017Ellipsoidal,asselborn2013control}. These sets intuitively provide an outer approximation on the next state of a system considering all possible realizations of uncertainties when applying a controller to the system at a given set-valued input.
We briefly review some of these properties and refer to \cite{Kurzhanskii1997Ellipsoidal} for an exhaustive introduction to ellipsoids and to the derivations for the following properties.

We use the basic definition of an ellipsoid,
\begin{equation}
\label{eq:found:ell:def}
E(p,Q) := \{ x \in \R^n | (x-p)^\mathrm{T} Q^{-1}(x-p) \leq 1\},
\end{equation}
with  \textit{center} $p \in \R^n$ and a symmetric positive definite (s.p.d) \textit{shape matrix} $Q \in \R^{n \times n}$.
Ellipsoids are invariant under \textit{affine subspace transformations} such that for $A \in \R^{r \times n}, r \leq n$ with full row rank and $b \in \R^r$, we have that
\begin{equation}
\label{multi_step:ell:affine_trafo}
A \cdot E(p,Q) + b = E(Ap+b,AQA^\mathrm{T}).
\end{equation}
The \textit{Minkowski sum} $E(p_1,Q_1) \oplus E(p_2,Q_2)$, i.e. the pointwise sum between two arbitrary ellipsoids, is in general not an ellipsoid anymore, but we have that
\begin{equation}
\label{multi_step:ell:minkowski_ellipsoids}
E(p_1,Q_1) \oplus E(p_2,Q_2) \subset E \big( p_1 + p_2, (1+c^{-1})Q_1 + (1+c)Q_2 \big)
\end{equation}
for all $c > 0$. Moreover, the minimizer of the trace of the resulting shape matrix is analytically given as $c = \sqrt{Tr(Q_1) / Tr(Q_2)}$. A particular problem that we encounter is finding the maximum distance $r$ to the center of an ellipsoid $E:= E(0,Q)$ under a special transformation, i.e.
\begin{align}
\label{found:ell:prob:rayleigh_variant}
r(Q,S) = \max_{x \in E(p,Q)} ||S(x-p)||_2 =  \max_{s^\mathrm{T} Q^{-1}s \leq 1} s^\mathrm{T} S^\mathrm{T} S s,
\end{align}
where $S \in \R^{m \times n}$ with full column rank. This is a generalized eigenvalue problem of the pair $(Q,S^TS)$ and the optimizer is given as the square-root of the largest generalized eigenvalue.

\section{Safe Model Predictive Control}
\label{sec:safempc}

In this section, we use the assumptions in~\cref{ps} to design a control scheme that fulfills our safety requirements in \cref{ps:def:epsilon_safety}. We construct reliable, multi-step ahead predictions based on our GP model and use MPC to actively optimize over these predicted trajectories under safety constraints. Using \cref{ps:ass:safe_set}, we use a terminal set constraint to theoretically prove the safety of our method.

\subsection{Multi-step Ahead Predictions}
From~\cref{gp:lemma:rkhs_confidence_ivals} and our prior model~$h(x_t,u_t)$, we directly obtain high-probability confidence intervals on~$f(x_t,u_t)$ uniformly for all $t \in \N$. We extend this to over-approximate the system after a sequence of inputs $(u_t,u_{t+1},..)$. The result is a sequence of set-valued confidence regions that contain the true dynamics of the system with high probability.

\paragraph{One-step ahead predictions}
We compute an ellipsoidal confidence region that contains the next state of the system with high probability when applying a control input, given that the current state is contained in an ellipsoid.
In order to approximate the system, we linearize our prior model $h(x_t,u_t)$ and use the affine transformation property \cref{multi_step:ell:affine_trafo} to compute the ellipsoidal next state of the linearized model. Next, we approximate the unknown model-error $g(x_t,u_t)$ using the confidence intervals of our GP model. We finally apply Lipschitz arguments to outer-bound the approximation errors. We sum up these individual approximations, which result in an ellipsoidal approximation of the next state of the system. This is illustrated in~\cref{main:multi_step:fig:one_step_overapproximation}. We formally derive the necessary equations in the following paragraphs. The reader may choose to skip the technical details of these approximations, which result in~\cref{main:multi_step:lemma:one_step}.

\begin{figure}[t]
\includegraphics{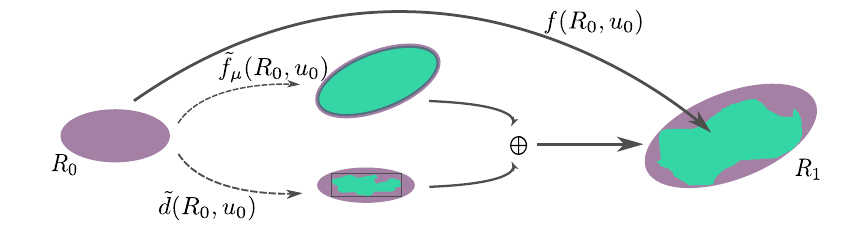}
\caption{Decomposition of the over-approximated image of the system \cref{ps:eq:system} under an ellipsoidal input $R_0$. The exact, unknown image of $f$ (right, green area) is approximated by the linearized model $\tilde{f}_\mu$ (center, top) and the remainder term $\tilde{d}$, which accounts for the confidence interval and the linearization errors of the approximation (center, bottom). The resulting ellipsoid $R_1$ is given by the Minkowski sum of the two individual approximations.}
\label{main:multi_step:fig:one_step_overapproximation}
\end{figure}

We first regard the system $f$ in \cref{ps:eq:system}
for a single input vector $z = (x,u), \,f(z) = h(z) + g(z)$.
We linearly approximate $f$ around $\bar{z} = (\bar{x},\bar{u})$ via
\begin{equation}
\label{main:multi_step:eq:taylor}
f(z) \approx h(\bar{z}) + J_{h}(\bar{z}) (z-\bar{z}) + g(\bar{z}) = \tilde{f}(z),
\end{equation}
where $J_{h}(\bar{z}) = [A,B]$ is the Jacobian of $h$ at $\bar{z}$.

Next, we use the Lagrangian remainder theorem~\cite{Breiman1993Deterministic} on the linearization of $h$ and apply a continuity argument on our locally constant approximation of $g$. This results in an upper-bound on the approximation error,
\begin{equation}
\label{main:eq:linearization_approx}
|f_j(z) - \tilde{f}_j(z) | \leq \frac{L_{\nabla h,j}}{2} ||z - \bar{z}||_2^2 + L_g ||z-\bar{z}||_2,
\end{equation}
where~$f_j(z)$ is the $i$th component of~$f$, $ 1\leq j \leq n_x$, $L_{\nabla h,j}$ is the Lipschitz constant of the gradient $\nabla h_j$, and $L_g$ is the Lipschitz constant of~$g$, which exists by~\cref{gp:lemma:lipschitz}.

The function~$\tilde{f}$ depends on the unknown model error~$g$. We approximate~$g$ with the statistical GP model, $\mu_n(\bar{z}) \approx g(\bar{z})$. From \cref{gp:lemma:rkhs_confidence_ivals} we have
\begin{equation}
\label{main:eq:rkhs_bound}
|g_j(\bar{z})-\mu_{n,j}(\bar{z})| \leq \beta_n \sigma_{n,j}(\bar{z}), \, 1\leq j \leq n_x,
\end{equation}
with high probability.
We combine \cref{main:eq:linearization_approx} and \cref{main:eq:rkhs_bound} to obtain
\begin{equation}
\label{main:eq:errorbound_deterministic}
| f_j(z) - \tilde{f}_{\mu,j}(z)| \leq \beta_n \sigma_j(\bar{z}) + \frac{L_{\nabla h,j}}{2} ||z - \bar{z}||_2^2 + L_g ||z-\bar{z}||_2,
\end{equation}
where $1 \leq j \leq n_x$ and
$
\tilde{f}_\mu(z) = h(\bar{z}) + J_{h}(\bar{z}) (z-\bar{z}) + \mu_n(\bar{z}).
$
We can interpret  \cref{main:eq:errorbound_deterministic} as the edges of the confidence hyper-rectangle
\begin{equation}
\label{main:eq:model:deterministic}
\tilde{m}(z)  = \tilde{f}_\mu(z) \pm [\beta_n\sigma_{n-1}(\bar{z})+\frac{L_{\nabla h}}{2} ||z - \bar{z}||_2^2 + L_g ||z-\bar{z}||_2 ],
\end{equation}
where $L_{\nabla h} = [L_{\nabla h,1},..,L_{\nabla h,n_x}]$ and we use the shorthand notation $a \pm b := [a_1 \pm b_1] \times [a_{n_x} \pm b_{n_x}],\, a,b \in \R^{n_x}$.

We are now ready to compute a confidence region based on an ellipsoidal state $R = E(p,Q) \subset \R^{n_x}$ and a fixed input $u \in \R^{n_u}$, by over-approximating the output of the system $f(R,u) = \{ f(x,u)| x \in R\}$ for ellipsoidal inputs $R$. Here, we choose $p$ as the linearization center of the state and choose $\bar{u} = u$, i.e. $\bar{z} = (p,u)$. Since the function $\tilde{f}_\mu$ is affine, we can make use of \cref{multi_step:ell:affine_trafo} to compute
\begin{equation}
\tilde{f}_\mu(R,u) = E(h(\bar{z}) + \mu(\bar{z}),AQA^\mathrm{T}),
\end{equation}
resulting again in an ellipsoid.
This is visualized in \cref{main:multi_step:fig:one_step_overapproximation} by the upper ellipsoid in the center.
To upper-bound the confidence hyper-rectangle on the right hand side of \cref{main:eq:model:deterministic}, we upper-bound the term~$\| z - \bar{z} \|_2$ by
\begin{equation}
\label{main:eq:remainder_overapproximation_setinput}
l(R,u) := \max\limits_{\substack{z(x) = (x,u), \\ x \in R}} ||z(x) - \bar{z}||_2,
\end{equation}
which leads to
\begin{equation}
\label{main:eq:d_overapprox_rectangle}
\tilde{d}(R, u)  := \beta_n\sigma_{n-1}(\bar{z}) + L_{\nabla h} l^2(R,u) / 2 + L_g l(R,u).
\end{equation}
Due to our choice of $z, \bar{z}$, we have that $||z(x) - \bar{z}||_2 = ||x-p||_2$ and  we can use \cref{found:ell:prob:rayleigh_variant} to get
$l(R,u) = r(Q,I_{n_x}),$ which corresponds to the largest eigenvalue of $Q^{-1}$.
Using \cref{main:eq:remainder_overapproximation_setinput}, we can now over-approximate the right side of \cref{main:eq:model:deterministic} for inputs $R$ by an ellipsoid
\begin{equation}
0 \pm \tilde{d}(R,u) \subset E(0,Q_{\tilde{d}}(R,u)),
\end{equation}
where we obtain $Q_{\tilde{d}}(R,u)$  by over-approximating the hyper-rectangle~$\tilde{d}(R, u)$ with the ellipsoid~$E(0,Q_{\tilde{d}}(R,u))$ through $a \pm b \subset E(a,\sqrt{n_x}\cdot \mathrm{diag}([b_1,..,b_{n_x}])), \, \forall a,b \in R^{n_x}$. This is illustrated in \cref{main:multi_step:fig:one_step_overapproximation} by the lower ellipsoid in the center.
Combining the previous results, we can compute the final over-approximation using \cref{multi_step:ell:minkowski_ellipsoids},
\begin{equation}
\label{main:multi_step:eq:one_step_ahead_set}
R_+ = \tilde{m}(R,u) = \tilde{f}_\mu(R,u) \oplus  E(0,Q_{\tilde{d}}(R,u)).
\end{equation}
Since we carefully incorporated all approximation errors and extended the confidence intervals around our model predictions to set-valued inputs, we get the following generalization of  \cref{gp:lemma:rkhs_confidence_ivals}.
\begin{lemma}
\label{main:multi_step:lemma:one_step}
Let $\delta \in (0,1]$ and choose $\beta_n$ as in \cref{gp:lemma:rkhs_confidence_ivals}. Then, with probability greater than $1-\delta$, we have that:
\begin{equation}
\forall x \in R: f(x,u) \in \tilde{m}(R,u),
\end{equation}
uniformly for all $R = E(p,Q) \subset \mathcal{X}, \, u \in \mathcal{U}$.
\end{lemma}
\begin{proof}
Define $m(x,u) = h(x,u)+\mu_n(x,u) \pm \beta_n \sigma_{n-1}(x,u)$. From \cref{gp:lemma:rkhs_confidence_ivals} we have $\forall \, R \subset \mathcal{X},\, u \in \mathcal{U}$ that, with high probability, $\bigcup_{x \in R}f(x,u) \subset \bigcup_{x \in R} m(x,u)$. Due to the over-approximations, we have $\bigcup_{x \in R} m(x,u) \subset \tilde{m}(R,u)$.
\end{proof}

\cref{main:multi_step:lemma:one_step} allows us to compute confidence ellipsoid around the next state of the system, given that the current state of the system is given through an ellipsoidal \textit{belief}.

\paragraph{Multi-step ahead predictions}
We now use the previous results to compute a sequence of ellipsoids that contain a trajectory of the system with high-probability, by iteratively applying the one-step ahead predictions~\cref{main:multi_step:eq:one_step_ahead_set}.

Given an initial ellipsoid $R_0 \subset \R^{n_x}$ and control input $u_t \in \mathcal{U}$, we iteratively compute confidence ellipsoids as
\begin{equation}
\label{multi_step:eq:t_step_delta_reach}
R_{t+1} = \tilde{m}(R_t,u_t).
\end{equation}
We can directly apply \cref{main:multi_step:lemma:one_step} to get the following result.
\begin{corollary}
\label{multi_step:lemma:t_step_delta_reach}
Let $\delta \in (0,1]$ and choose $\beta_n$ as in \cref{gp:lemma:rkhs_confidence_ivals}. Choose $x_0 \in R_0 \subset \mathcal{X}$. Then the following holds jointly for all $t \geq 0$ with probability at least $1-\delta$:
$x_t \in R_t$,
where $z_t = (x_t,u_t) \in \mathcal{X} \times \mathcal{U}$, $R_0,R_1,..$ is computed as in \cref{multi_step:eq:t_step_delta_reach} and $x_t$ is the state of the system \cref{ps:eq:system} at time step $t$.
 \end{corollary}
 \begin{proof}
 Since \cref{main:multi_step:lemma:one_step} holds uniformly for all ellipsoids $R \subset \mathcal{X}$ and $u \in \mathcal{U},$ this is a special case that holds uniformly for all control inputs $u_t, \, t \in \N$ and for all ellipsoids $R_t,\, t \in \N$  obtained through \cref{multi_step:eq:t_step_delta_reach}.
 \end{proof}

\cref{multi_step:lemma:t_step_delta_reach} guarantees that, with high probability, the system is always contained in the propagated ellipsoids \cref{multi_step:eq:t_step_delta_reach}. Thus, if we provide safety guarantees for these sequences of ellipsoids, we obtain high-probability safety guarantees for the system ~\cref{ps:eq:system}.

\paragraph{Predictions under state-feedback control laws}
When applying multi-step ahead predictions under a sequence of feed-forward inputs $u_t \in \mathcal{X}$, the individual sets of the corresponding reachability sequence can quickly grow unreasonably large. This is because these \textit{open loop} input sequences do not account for future control inputs that could correct deviations from the model predictions.
Hence, we extend \cref{main:multi_step:eq:one_step_ahead_set} to
\textit{affine state-feedback control laws} of the form
\begin{equation}
\label{main:multi_sep:eq:feedback_ctrl}
\pi_t(x_t) := K_t(x_t-p_t) + u_t,
\end{equation}
where $K_t \in \R^{n_u \times n_x}$ is a feedback matrix and $u_t \in \R^{n_u}$ is the open-loop input. The parameter $p_t$ is determined through the center of the current ellipsoid $R_t = E(p_t,Q_t)$. Given an appropriate choice of $K_t$, the control law actively contracts the ellipsoids towards their center.
Similar to the derivations \cref{main:multi_step:eq:taylor}-\cref{main:multi_step:eq:one_step_ahead_set}, we can compute the function $\tilde{m}$ for affine feedback controllers \cref{main:multi_sep:eq:feedback_ctrl} $\pi_t$ and ellipsoids $R_t =E(p_t,Q_t)$.
The resulting ellipsoid is
\begin{equation}
\label{main:multi_step:eq:over_approx_state_feedback}
\tilde{m}(R_t,\pi_t) = E(h(\bar{z}_t)+\mu(\bar{z}_t),H_tQ_tH_t^\mathrm{T}) \oplus E(0,Q_{\tilde{d}}(R_t,\pi_t)),
\end{equation}
where $\bar{z_t} = (p_t,u_t)$ and $H_t = A_t + B_t K_t$. The set $E(0,Q_{\tilde{d}}(R_t,\pi_t))$ is obtained similarly to \cref{main:eq:remainder_overapproximation_setinput} as the ellipsoidal over-approximation of
\begin{equation}
0 \pm [\beta_n \sigma(\bar{z}) + L_{\nabla h} \frac{l^2(R_t,S_t)}{2}  + L_gl(R_t,S_t)],
\end{equation}
with $S_t = [I_{n_x},K_t^\mathrm{T}]$ and $l(R_t,S_t) = \max_{x \in R_t}||S_t(z(x)-\bar{z_t})||_2$. The theoretical results of \cref{main:multi_step:lemma:one_step} and \cref{multi_step:lemma:t_step_delta_reach} directly apply to the case of the uncertainty propagation technique \cref{main:multi_step:eq:over_approx_state_feedback}.
\subsection{Safety constraints}
The derived multi-step ahead prediction technique provides a sequence of ellipsoidal confidence regions around trajectories of the true system $f$ through  \cref{multi_step:lemma:t_step_delta_reach}. We can guarantee that the system is safe by verifying that the computed confidence ellipsoids are contained inside the polytopic constraints \cref{ps:ass:state_constraits} and \cref{ps:ass:ctrl_constraits}.
That is, given a sequence of feedback controllers $\pi_t, \, t = {0,..,T-1}$ we need to verify
\begin{equation}
R_{t+1} \subset \mathcal{X}, \, \pi_t(R_t) \subset \mathcal{U}, \, t=0,..,T-1,
\end{equation}
where $(R_0,..,R_T)$ is given through \cref{multi_step:eq:t_step_delta_reach}.

Since our constraints are polytopes, we have that
$\mathcal{X} = \bigcap_{i=1}^{m_x} \mathcal{X}_i$,
$\mathcal{X}_i = \{x \in \R^{n_x} |  [H_x]_{i,\cdot}  x - h_i^x\leq 0 \},$
where $[H_x]_{i,\cdot}$ is the $i$th row of $H^x$.
We can now formulate the state constraints through the condition $R_t = E(p_t,Q_t) \subset \mathcal{X}$ as $m_x$ individual constraints $R_t \subset \mathcal{X}_i,\, i= 1,..,m_x$, for which an analytical formulation exists~\cite{Hessem2002Closedloop},
\begin{equation}
\label{multi_step:ell:ell_polytope_state_constr}
[H_x]_{i,\cdot}p_t + \sqrt{[H_x]_{i,\cdot}Q_t[H_x]_{i,\cdot}^T}   \leq h^x_i, \, \forall i \in \{1,..,m_x\}.
\end{equation}

Moreover, we can use the fact that $\pi_t$ is affine  in~$x$ to obtain $\pi_t(R_t) = E(k_t,K_tQ_t,K_t^T)$, using \cref{multi_step:ell:affine_trafo}. The corresponding control constraint $\pi_t(R_t) \subset \mathcal{U}$ is then equivalently given by
\begin{equation}
\label{multi_step:ell:ell_polytope_control_constr}
[H_u]_{i,\cdot}u_t + \sqrt{[H_u]_{i,\cdot}K_tQ_tK_t^\mathrm{T} [H_u]_{i,\cdot}^\mathrm{T}} \leq h^u_i, \, \forall i \in \{1,..,m_u\}.
\end{equation}

\subsection{The SafeMPC algorithm}

Based on the previous results, we formulate a MPC scheme that optimizes the long-term performance of our system, while satisfying the safety condition in~\cref{ps:def:epsilon_safety}:
\begin{subequations}
\label{main:safempc:mpc:mpc_problem}
\begin{align}
&\underset{\pi_0,..,\pi_{T-1}}{\text{minimize}}
& & J_t(R_0,..,R_T) \\
\label{main:safempc:mpc:mpc_problem:a}
& \text{subject to }
 & & R_{t+1} = \tilde{m}(R_t,\pi_t), \, t=0,..,T-1   \\
& & & R_t \subset \mathcal{X}, \, t = 1,..,T-1 \\
& & & \pi_t(R_t) \subset \mathcal{U}, t = 0,..,T-1 \\
\label{main:safempc:mpc:mpc_problem:d}
& & & R_T \subset \mathcal{X}_{\mathrm{safe}},
\end{align}
\end{subequations}
where $R_0 := \{x_t\}$ is the current state of the system and the intermediate state and control constraints are defined in \cref{multi_step:ell:ell_polytope_state_constr}, \cref{multi_step:ell:ell_polytope_control_constr}. The terminal set constraint $R_T \subset \mathcal{X}_{\mathrm{safe}}$ has the same form as \cref{multi_step:ell:ell_polytope_state_constr} and can be formulated accordingly. The objective $J_t$ can be chosen to suit the given control task.

Due to the terminal constraint $R_T \subset \mathcal{X}_{\mathrm{safe}}$, a solution to  \cref{main:safempc:mpc:mpc_problem} provides a sequence of feedback controllers $\pi_0,..,\pi_T$ that steer the system back to the safe set $\mathcal{X}_{\mathrm{safe}}$.
We cannot directly show that a solution to MPC problem \cref{main:safempc:mpc:mpc_problem} exists at every time step (this property is known as recursive feasibility) without imposing additional assumption, e.g. on the safety controller $\pi_{\mathrm{safe}}$. However, employing a control scheme similar to standard robust MPC, we guarantee that such a sequence of feedback controllers exists at every time step as follows: Given a feasible solution $ \Pi_t = ( \pi_t^0,.., \pi_{t}^{T-1} )$ to \cref{main:safempc:mpc:mpc_problem} at time $t$, we apply the first feed-back control $\pi_t^0$. In case we do not find a feasible solution to \cref{main:safempc:mpc:mpc_problem} at the next time step, we shift the previous solution in a receding horizon fashion and append $\pi_{\mathrm{safe}}$ to the sequence to obtain $\Pi_{t+1} = ( \pi_t^1,.., \pi_{t}^{T-1}, \pi_{\mathrm{safe}} )$. We repeat this process until a new feasible solution exists that replaces the previous input sequence. This procedure is summarized in \cref{main:safempc:algo:safety_algorithm}.
We now state the main result of the paper that guarantees the safety of our system under the proposed algorithm.
\begin{theorem}
\label{main:safempc:theorem:safe_algorithm}
Let $\pi$ be the controller defined through algorithm \cref{main:safempc:algo:safety_algorithm} and $x_0 \in X_{\mathrm{safe}}$.
Then the system \cref{ps:eq:system} is $\delta-$safe under the controller $\pi$.
\end{theorem}
\begin{proof} %
From \cref{multi_step:lemma:t_step_delta_reach}, the ellipsoidal outer approximations (and by design of the MPC problem, also the constraints \cref{ps:ass:state_constraits}) hold uniformly with high probability for \textit{all} closed-loop systems $f_\Pi$, where $\Pi$ is a feasible solution to \cref{main:safempc:mpc:mpc_problem}, over the corresponding time horizon $T$.
Hence we can show uniform high probability safety by induction.
Base case: If \cref{main:safempc:mpc:mpc_problem} is infeasible, we are $\delta$-safe using the backup controller $\pi_{\mathrm{safe}}$ of \cref{ps:ass:safe_set}, since $x_0 \in \mathcal{X}_\mathrm{safe}$. Otherwise the controller returned from \cref{main:safempc:mpc:mpc_problem} is $\delta$-safe as a consequence of \cref{multi_step:lemma:t_step_delta_reach} and the terminal set constraint that leads to $x_{t+T} \in \mathcal{X}_\mathrm{safe}$.
Induction step: let the previous controller $\pi_t$ be $\delta$-safe. At time step $t+1$, if \cref{main:safempc:mpc:mpc_problem} is infeasible then $\Pi_t$ leads to a state $x_{t+T} \in \mathcal{X}_\mathrm{safe}$, from which the backup-controller is $\delta$-safe by \cref{ps:ass:safe_set}. If \cref{main:safempc:mpc:mpc_problem} is feasible, then the return path is $\delta$-safe by \cref{multi_step:lemma:t_step_delta_reach}.
\end{proof}

\begin{algorithm}[t]
\caption{Safe Model Predictive Control (\textsc{SafeMPC})}
\label{main:safempc:algo:safety_algorithm}
\begin{algorithmic}[1]
\State \textbf{Input:} Safe policy $\pi_\mathrm{safe}$, dynamics model $h$, statistical model $\mathcal{GP}(0,k)$. 
\State $\Pi_0 \gets \{\pi_{\mathrm{safe}},..,\pi_{\mathrm{safe}}\}$ with $|\Pi_0| = T$
\For{$t=0,1,..$}
	\State $J_t \gets$ objective from high-level planner
	\State feasible, $\Pi \gets$ solve MPC problem \cref{main:safempc:mpc:mpc_problem}
    \If{feasible}: $\Pi_t \gets \Pi$
    \Else: $\Pi_t \gets \left(\Pi_{t-1,1:T-1},\pi_{\mathrm{safe}}  \right)$ 
    \EndIf
    \State $x_{t+1} \gets$ apply $u_t = \Pi_{t,0}(x_t)$ to the system \cref{ps:eq:system}
\EndFor
\end{algorithmic}%
\end{algorithm}


\subsection{Optimizing long-term behavior}
\label{main:perf_trajectories}
While the proposed MPC problem \cref{main:safempc:mpc:mpc_problem} yields a safe return strategy, we are often interested in a controller that optimizes performance over a possibly much longer horizon. In the autonomous driving example, a safety trajectory that stabilizes the car towards the center of the lane can be much shorter than for planning a steering maneuver before entering a turn.
We hence propose to simultaneously plan a \textit{performance trajectory} $s_0,..,s_H$ under a sequence of inputs $\pi^{\textrm{perf}}_{0},..,\pi^{\textrm{perf}}_{H-1}$ using a performance-model $m_{\mathrm{perf}}$ along with the return strategy that we obtain when solving \cref{main:safempc:mpc:mpc_problem}. We do not make any assumptions on the performance model which could be given by one of the approximate uncertainty propagation methods proposed in the literature (see, e.g. \cite{Hewing2018Cautious} for an overview). In order to maintain the safety of our system, we enforce that the first $r \in \{1,.,  \min\{T,H\} \}$ controls are the same for both trajectories, i.e. we have that $\pi_k = \pi^{\mathrm{perf}}_k, k = 0,..,r-1 $. This extended MPC problem is
\begin{equation}
\label{main:safempc:perf:mpc_problem}
\begin{aligned}
&\underset{\substack{\pi_t,..,\pi_{t+T-1} \\ \pi^{\mathrm{perf}}_{t},..,\pi^{\mathrm{perf}}_{t+H-1}}}{\text{minimize}}
& & J_t(s_t,..,s_{t+H}) \\
& \text{subject to }
 & & \cref{main:safempc:mpc:mpc_problem:a}-\cref{main:safempc:mpc:mpc_problem:d},\,  t = 0,..,T-1 \\
& & & s_{t+1}  = m_{\mathrm{perf}}(s_t,\pi^{\mathrm{perf}}_{t}), t = 0,..,H-1 \\
& & & \pi_{t} = \pi^{\mathrm{perf}}_{t}, \, t = 0,..,r-1 ,
\end{aligned}
\end{equation}
where we replace \cref{main:safempc:mpc:mpc_problem} with this problem in  \cref{main:safempc:algo:safety_algorithm}. The safety guarantees of  \cref{main:safempc:theorem:safe_algorithm} directly translate to this setting, since we can always fall back to the return strategy.

\subsection{Discussion}
Algorithm \cref{main:safempc:algo:safety_algorithm} theoretically guarantees that the system remains safe, while actively optimizing for performance via the MPC problem $\cref{main:safempc:perf:mpc_problem}$. This problem can be solved by commonly used, nonlinear programming (NLP) solvers, such as the \textit{Interior Point OPTimizer (Ipopt, \cite{Wachter2006Implementation})}. Due to the solution of the eigenvalue problem \cref{found:ell:prob:rayleigh_variant} that is required to compute \cref{main:multi_step:eq:one_step_ahead_set}, our uncertainty propagation scheme is not analytic. However, we can still obtain exact function values and derivative information by means of algorithmic differentiation, which is at the core of many state-of-the-art optimization software libraries~\cite{Andersson2013b}.

One way to further reduce the conservatism of the multi-step ahead predictions is to linearize the GP mean prediction $\mu_n(x_t,u_t)$, which we omitted for clarity.
\section{Experiments}
In this section, we evaluate the proposed \textsc{SafeMPC} algorithm to safely explore the dynamics of an inverted pendulum system.

\begin{figure*}[t]
\includegraphics{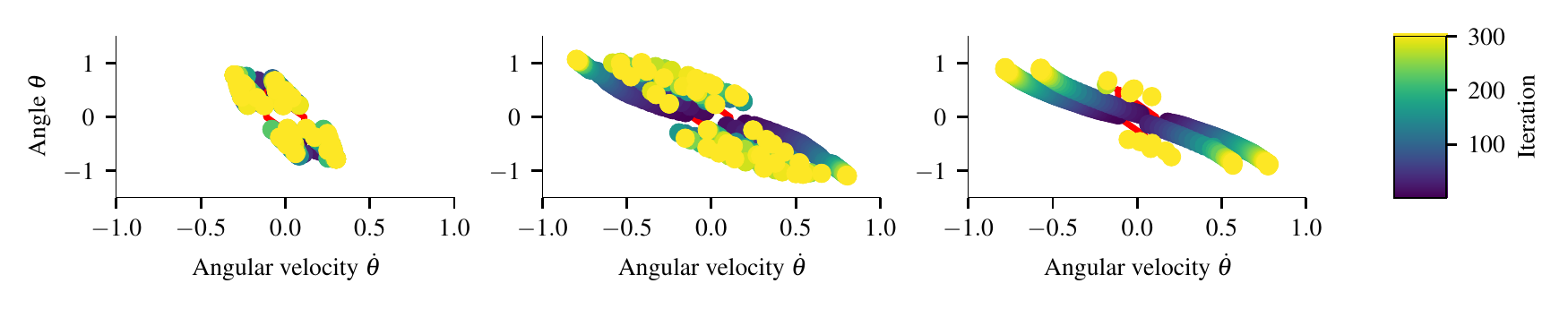}
\caption{Visualization of the samples acquired in the static exploration setting in~\cref{ssec:static_exploration} for $T \in \{1,4,5\}$. The algorithm plans informative paths to the safe set $\mathcal{X}_{\mathrm{safe}}$ (red polytope in the center). The baseline sample set for $T=1$ (left) is dense around origin of the system. For $T=4$ (center) we get the optimal trade-off between cautiousness due to a long horizon and limited length of the return trajectory due to a short horizon. The exploration for $T=5$ (right) is too cautious, since the propagated uncertainty at the final state is too large.}
\label{exp:fig:static_exploration:sampleset}
\end{figure*}

The continuous-time dynamics of the pendulum are given by $ml^2\ddot{\theta} = gml\sin(\theta) - \eta \dot{\theta} + u$, where $m = 0.15 \unit{kg}$ and $l=0.5\unit{m}$ are the mass and length of the pendulum, respectively, $\eta = 0.1 \unitfrac{Nms}{rad}$ is a friction parameter, and $g = 9.81 \unitfrac{m}{s^2}$ is the gravitational constant. The state of the system $x = (\theta,\dot{\theta})$ consists of the angle $\theta$ and angular velocity $\dot{\theta}$ of the pendulum. The system is controlled by a torque~$u$ that is applied to the pendulum. The origin of the system corresponds to the pendulum standing upright.

The system is underactuated with control constraints $\mathcal{U} = \{ u \in \R |-1 \leq u \leq 1\}$. Due to these limits, the pendulum becomes unstable and falls down beyond a certain angle. We do not impose state constraints,~$\mathcal{X} = \R^2$. However the terminal set constraint \cref{main:safempc:mpc:mpc_problem:d} of the MPC problem \cref{main:safempc:mpc:mpc_problem} acts as a stability constraint and prevents the pendulum from falling.
Apart from being smooth, we do not make any assumptions on our prior model $h$ and we choose it to be a linearized and discretized approximation to the true system with a lower mass and neglected friction as in ~\cite{Berkenkamp2017Safe-1}.
The safety controller~$\pi_{\mathrm{safe}}$ is a discrete-time, infinite horizon linear quadratic regulator (LQR,\cite{Kwakernaak1972Linear}) of the approximated system $h$ with cost matrices
$ Q = \mathrm{diag}([1,2])$, $R = 20$. The corresponding safety region $\mathcal{X}_{\mathrm{Safe}}$ is given by a conservative polytopic inner-approximation of the true region of attraction of $\pi_{\mathrm{safe}}$.
We use the same mixture of linear and Mat\'ern kernel functions for both output dimensions, albeit with different hyperparameters. We initially train our model with a dataset $(Z_0,\tilde{y}_0)$ sampled inside the safe set using the backup controller $\pi_{\mathrm{Safe}}$. That is, we gather $n_0 = 25$ initial samples $Z_0 = \{ z_1^0,..,z^0_{n_0} \}$ with $z_i^0 = (x_i, \pi_{\mathrm{safe}}(x_i)), x_i \in \mathcal{X}_{\mathrm{safe}}, \, i = 1,..,n$ and observed next states $ \tilde{y}_0 =\{y^0_0,..,y^0_{n_0}\} \subset \mathcal{X}_{\mathrm{safe}}$. The theoretical choice of the scaling parameter $\beta_n$ for the confidence intervals in \cref{gp:lemma:rkhs_confidence_ivals} can be conservative and we choose a fixed value of $\beta_n = 2$ instead, following~\cite{Berkenkamp2017Safe-1}.

We aim to iteratively collect the most informative samples of the system, while preserving its safety. To evaluate the exploration performance, we use the mutual information $I(g_{Z_n},g)$ between the collected samples $Z_n = \{z_0,..,z_n\} \cup Z_0$ and the GP prior on the unknown model-error~$g$, which can be computed in closed-form~\cite{Srinivas2010Gaussian}.

\subsection{Static Exploration}
\label{ssec:static_exploration}

For a first experiment, we assume that the system is \textit{static}, so that we can reset the system to an arbitrary state $x_n \in \R^2$ in every iteration. In the static case and without terminal set constraints, a provably close-to-optimal exploration strategy is to, at each iteration~$n$, select state-action pair~$z_{n+1}$ with the largest predictive standard deviation~\cite{Srinivas2010Gaussian}
\begin{equation}
\label{exp:max_gp_confidence_intervals}
z_{n+1} = \underset{z \in \mathcal{X} \times \mathcal{U}}{\arg\max} \sum_{1 \leq j \leq n_x } \sigma_{n,j}(z),
\end{equation}
where $\sigma_{n,j}^2(\cdot)$ is the predictive variance~\cref{found:gp:pred_var} of the $j$th $\mathcal{GP}(0,k_j)$ at the $n$th iteration. Inspired by this, at each iteration we collect samples by solving the MPC problem \cref{main:safempc:mpc:mpc_problem} with cost function~$J_n = -\sum_{j=1}^{n_x} \sigma_{n,j}$, where we additionally optimize over the initial state $x_n \in \mathcal{X}$. Hence, we visit high-uncertainty states, but only allow for state-action pairs $z_n$ that are part of a feasible return trajectory to the safe set $\mathcal{X}_{\mathrm{Safe}}$.

Since optimizing the initial state is highly non-convex, we solve the problem iteratively with $25$ random initializations to obtain a good approximation of the global minimizer. After every iteration, we update the sample set $Z_{n+1} = Z_n \cup \{z_n\}$, collect an observation $(z_n,y_n)$ and update the GP models. We apply this procedure for varying horizon lengths.

The resulting sample sets are visualized for varying horizon lengths $T \in \{1,..,5\}$ with $300$ iterations in \cref{exp:fig:static_exploration:sampleset}, while \cref{exp:fig:static_exploration:inf_gain} shows how the mutual information of the sample sets $Z_i,\, i = 0,..,n$ for the different values of $T$. For short time horizons ($T=1$), the algorithm can only slowly explore, since it can only move one step outside of the safe set. This is also reflected in the mutual information gained, which levels off quickly. For a horizon length of $T=4$, the algorithm is able to explore a larger part of the state-space, which means that more information is gained. For larger horizons, the predictive uncertainty of the final state is too large to explore effectively, which slows down exploration initially, when we do not have much information about our system. The results suggest that our approach could further benefit from adaptively choosing the horizon during operation, e.g. by employing a variable horizon MPC approach \cite{RichardsArthur2006Robust}, or by increasing the horizon when the mutual information saturates for the current horizon.
\begin{figure}[t]
\includegraphics{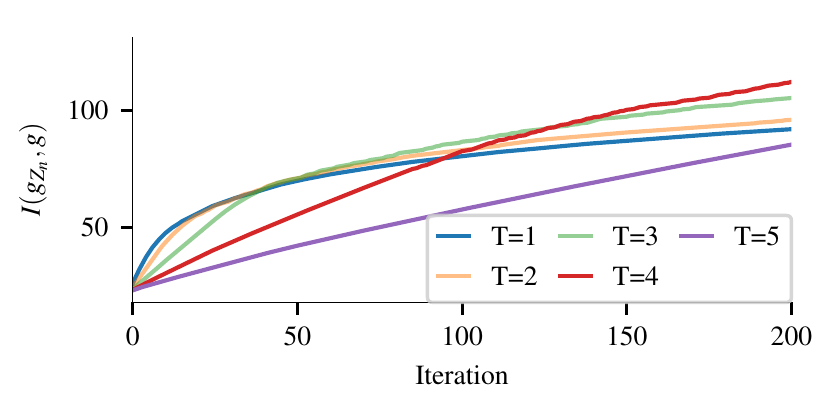}
\caption{Mutual information $I(g_{Z_n},g), \, n=1,..,200$ for horizon lengths $T \in \{1,..,5\}$. Exploration settings with shorter horizon gather more informative samples at the beginning, but less informative samples in the long run. Longer horizon lengths result in less informative samples at the beginning, due to uncertainties being propagated over long horizons. However, after having gathered some knowledge they quickly outperform the smaller horizon settings. The best trade off is found for $T=4$.
}
\label{exp:fig:static_exploration:inf_gain}
\end{figure}

\subsection{Dynamic Exploration}

As a second experiment, we collect informative samples during operation; without resetting the system at every iteration. Starting at~$x_0 \in \mathcal{X}_\mathrm{safe}$, we apply the \textsc{SafeMPC}, \cref{main:safempc:algo:safety_algorithm}, over $200$ iterations.
We consider two settings. In the first, we solve the MPC problem~\cref{main:safempc:mpc:mpc_problem} with $-J_n$ given by \cref{exp:max_gp_confidence_intervals}, similar to the previous experiments. In the second setting, we additionally plan a performance trajectory as proposed in~\cref{main:perf_trajectories}. We define the states of the performance trajectory as Gaussians $s_{t} = \mathcal{N} (m_t,S_t) \in \R^{n_x} \times \R^{n_x \times n_x}$ and the next state is given by the predictive mean and variance of the current state $m_t$ and applied action $u_t$. That is,~$s_{t+1}= \mathcal{N}(m_{t+1},S_{t+1})$ with
\begin{equation*}
m_{t+1} = \mu_n(m_t,u_t), \, S_{t+1} = \Sigma_n(m_t,u_t) , t = 0,..,H-1,
\end{equation*}
where $\Sigma_n(\cdot) = \mathrm{diag}(\sigma_n^2(\cdot))$ and $m_0 = x_n$.
This simple approximation technique is known as \textit{mean-equivalent} uncertainty propagation. We define the cost-function $-J_t = \sum_{t=0}^{H} \mathrm{trace}(S_t^{1/2}) - \sum_{t=1}^{T} (m_t-p_t)^TQ_{\mathrm{perf}}(m_t-p_t)$, which maximizes the sum of predictive confidence intervals along the trajectory $s_1,..,s_H$, while penalizing deviation from the safety trajectory. We choose~$r = 1$ in the problem \cref{main:safempc:perf:mpc_problem}, i.e. the first action of the safety trajectory and performance trajectory are the same. As in the static setting, we update our GP models after every iteration.

We evaluate both settings for varying $T \in \{1,..,5\}$ and fixed $H=5$ in terms of their mutual information in~\cref{exp:fig:dynamic_expl:information_gain}. We observe a similar behavior as in the static exploration experiments and get the best exploration performance for $T=4$ with a slight degradation of performance for $T=5$. We can see that, except for $T=1$, the performance trajectory decomposition setting consistently outperforms the standard setting. Planning a performance trajectory (green) provides the algorithm with an additional degree of freedom, which leads to drastically improved exploration performance.

\begin{figure}[t]
\includegraphics{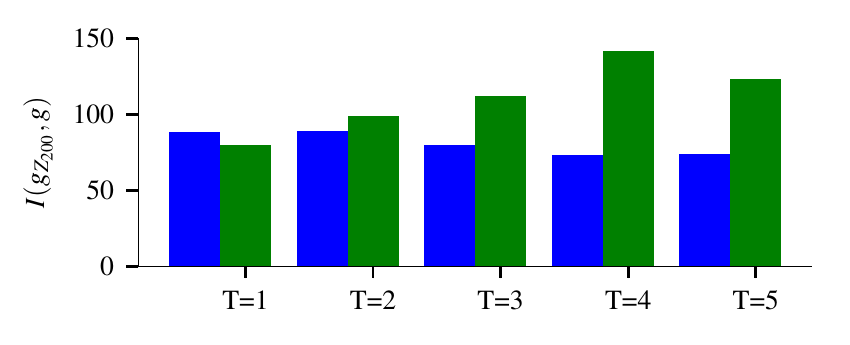}
\caption{Comparison of the information gathered from the system after $200$ iterations for the standard setting (blue) and the setting where we plan an additional performance trajectory (green). 
}
\label{exp:fig:dynamic_expl:information_gain}
\end{figure}

\section{Conclusion}
We introduced~\textsc{SafeMPC}, a learning-based MPC scheme that can safely explore partially unknown systems. The algorithm is based on a novel uncertainty propagation technique that uses a reliable statistical model of the system. As we gather more data from the system and update our statistical mode, the model becomes more accurate and control performance improves, all while maintaining safety guarantees throughout the learning process.

\bibliographystyle{IEEEtran}
\bibliography{refs_zotero}

\end{document}